\documentclass[12pt,a4paper,oneside,reqno]{amsart}

\usepackage{amsmath,amssymb}
\usepackage{graphicx,psfrag}

\newtheorem{thm}{Theorem}[section]
\newtheorem{prop}{Proposition}[section]

\newcommand{\ndash}{--}
\newcommand{\citep}{\cite}

\newcommand{\ax}[1]{\ensuremath{\mathsf{#1}}}

\newcommand{\Ph}{\mathsf{Ph}}
\newcommand{\Q}{\mathit{Q}}
\newcommand{\B}{\mathit{B}}
\newcommand{\W}{\mathsf{W}}
\newcommand{\IOb}{\mathsf{IOb}}
\newcommand{\Ob}{\mathsf{Ob}}
\newcommand{\vx}{\bar x}
\newcommand{\vy}{\bar y}
\newcommand{\vz}{\bar z}
\newcommand{\vw}{\bar w}
\newcommand{\then}{\;\rightarrow\;}
\newcommand{\oszt}{\slash}
\newcommand{\gyok}{\sqrt{\phantom{n}}}
\newcommand{\de}{\stackrel{d}{=}}
\newcommand{\defiff}{\ \stackrel{d}{\Longleftrightarrow}\ }
\renewcommand{\iff}{\leftrightarrow}

\newcommand{\ev}{\mathsf{ev}}

\newcommand{\w}{\mathsf{w}}

\usepackage{color}
\definecolor{rd}{rgb}{.7,0,0}
\definecolor{bl}{rgb}{0,0,1}
\definecolor{gr}{rgb}{0,.5,0}

\newcommand{\regi}[1]{} 

\begin{document}

\title{On Logical Analysis of Relativity Theories}
\author{Hajnal Andr\'eka, Istv\'an N\'emeti, Judit X. Madar\'asz and Gergely Sz\'ekely}
\thanks{This research is supported by the Hungarian Scientific Research Fund for basic
 research grant No.~T81188, as well as by a Bolyai grant
 for J.~X.~Madar\'asz.}

\begin{abstract}
The aim of this paper is to give an introduction to our axiomatic
logical analysis of relativity theories.
\end{abstract}
\maketitle

\section{introduction}

Our general aim is to build up relativity theories as theories in the
sense of mathematical logic. So we axiomatize relativity theories
within pure first-order logic (FOL) using simple, comprehensible and
transparent basic assumptions (axioms). We strive to prove all the
surprising predictions of relativity from a minimal number of
convincing axioms. We eliminate tacit assumptions from relativity by
replacing them with explicit axioms (in the spirit of the foundation
of mathematics and Tarski's axiomatization of geometry). We also
elaborate logical and conceptual analysis of our theories.

Logical axiomatization of physics, especially that of relativity
theory, is not a new idea, among others, it goes back to such
leading scientists as Hilbert, Reichenbach, Carnap, G{\"o}del, and
Tarski. Relativity theory was intimately connected to logic from the
beginning, it was one of the central subjects of logical positivism.
For a short survey on the broader literature,
 see, e.g., \citep{AMNsamples}. Our aims go beyond these
approaches in that along with axiomatizing relativity theories we also
analyze in detail their logical and conceptual structure and, in
general, investigate them in various ways (using our logical framework
as a starting point).

A novelty in our approach is that we try to keep the transition from
special relativity to general relativity logically transparent and
illuminating. We ``derive'' the axioms of general relativity from
those of special relativity in two natural steps.  First we extend our
axiom system for special relativity with accelerated observers
(sec.\ref{acc-sec}). Then we eliminate the distinguished status of
inertial observers at the level of axioms (sec.\ref{gen-sec}).

Some of the questions we study to clarify the logical structure of
relativity theories are:
\begin{itemize}
\item
 What is believed and why?
\item
Which axioms are responsible for certain predictions?
\item
What happens if we discard some axioms?
\item
Can we change the axioms and at what price?
\end{itemize}

Our aims stated in the first paragraph reflect, partly, the fact
that we axiomatize a physical theory. Namely, in physics the role of
axioms (the role of statements that we assume without proofs) is
more fundamental than in mathematics. Among others, this is why we
aim to formulate simple, logically transparent and intuitively
convincing axioms.Our goal is that on our approach, surprising or
unusual predictions be theorems and not assumed as axioms.
For example, the prediction ``no faster than light
motion ...'' is a theorem on our approach and not an axiom, see
Thm.\ref{thm-noftl}.

Getting rid of unnecessary axioms is especially important in a
physical theory. When we check the applicability of a physical theory
in a situation, we have to check whether the axioms of the theory hold
or not. For this we often use empirical facts (outcomes of concrete
experiments). However, these correspond to existentially
quantified theorems\footnote{We do not want to assume every
  experimental fact as an axiom. We only want them to be consequences
  of our theories.} rather than to universally quantified statements{\textemdash}which
the axioms usually are. Thus while we can easily disprove the axioms
by referring to empirical facts, we can verify these axioms only to
a certain degree. Some of the literature uses the term 'empirical
  fact' for universal generalization of an empirical fact elevated to
  the level of axioms, see, e.g., \cite[\S
  4]{GSz}, \citep{Szabo}. We simply call these generalizations
  (empirical) axioms.

\section{why relativity?}

For one thing, Einstein's theory of relativity not just had but still
has a great impact on many areas of science. It has also greatly
affected several areas in the philosophy of science. Relativity theory
has an impact even on our every day life, e.g., via GPS technology
(which cannot work without relativity theory). Any theory with such an
impact is also interesting from the point of view of axiomatic
foundations and logical analysis.

Since spacetime is a similar geometrical object as space,
axiomatization of relativity theories (or spacetime theories in
general) is a natural continuation of the works of Euclid, Hilbert,
Tarski and many others axiomatizing the geometry of space.

\section{why axiomatic method?}

There are many examples showing the benefits of using axiomatic
method.  For example, if we decompose relativity theories into little
parts (axioms), we can check what happens to our theory if we drop,
weaken or replace an axiom or we can take any prediction, such as the
twin paradox, and check which axiom is and which is not needed to
derive it.  This kind of reverse thinking helps to answer the why-type
questions.  For details on answering why-type questions by the
methodology of the present work, see \cite[12{\ndash}13.]{pezsgo},
\citep{wqp}.

The success story of axiomatic method in the foundations of
mathematics also suggests that it is worth applying this method in
the foundations of spacetime theories \citep{FriFOM1},
\citep{FriFOM2}. Let us note here that
Euclid's axiomatic-deductive approach to geometry also made a great
impression on the young Einstein, see
\citep{herschbach}.

Among others, logical analysis makes relativity theory modular: we can
change some axioms, and our logical machinery ensures that we can
continue working in the modified theory. This modularity might come
handy, e.g., when we want to unify general relativity and quantum
theory to a theory of quantum gravity.
For further reasons why to apply the axiomatic method to spacetime
theories, see, e.g.,  \citep{AMNsamples}, \citep{pezsgo}, \citep{guts},
\citep{schutz}, \citep{suppes}.

\section{why first-order logic?}

We aim to provide a logical foundation for spacetime theories similar
to the rather successful foundations of mathematics, which, for good
reasons, was performed strictly within FOL.  One of these reasons is
that FOL helps to avoid tacit assumptions.  Another is that FOL has a
complete inference system while second-order logic (or higher-order
logic) cannot have one.

Still another reason for choosing FOL is that it can be viewed as a
fragment of natural language with unambiguous syntax and
semantics. Being a {\it fragment of natural language} is useful in our
project because one of our aims is to make relativity theory
accessible to a broad audience. {\it Unambiguous syntax and semantics}
are important, because they make it possible for the reader to always
know what is stated and what is not stated by the axioms. Therefore
they can use the axioms without being familiar with all the tacit
assumptions and rules of thumb of physics (which one usually learns via
many, many years of practice).

For further reasons why to stay within FOL when dealing with axiomatic
foundations, see, e.g., \cite[\S Appendix: Why FOL?]{pezsgo},
\citep{ax}, \cite[\S 11]{Szphd}, \citep{vaananen}, \citep{wolenski}.

\section{special relativity}

Before we present our axiom system let us go back to Einstein's
original (logically non-formalized) postulates.  Einstein based his
special theory of relativity on two postulates, the principle of
relativity and the light principle: ``The laws by which the states of
physical systems undergo change are not affected, whether these
changes of state be referred to the one or the other of two systems of
coordinates in uniform translatory motion.'' and ``Any ray of light
moves in the `stationary' system of co-ordinates with the determined
velocity $c$, whether the ray be emitted by a stationary or by a
moving body.'', see \citep{einstein05}.

The logical formulation of Einstein's principle of relativity is not
an easy task since it is difficult to capture axiomatically what ``the
laws of nature'' are in general. Nevertheless, the principle of
relativity can be captured by our FOL approach, see \citep{pezsgo},
\cite[\S 2.8.3]{Mphd}.

Instead of formulating the two original principles,
we formulate the following consequence of theirs: ``the speed of light
signals is the same in every direction everywhere according to
every inertial observer'' (and not just according to the
`stationary' observer). Here we will base our axiomatization on this
consequence and call it light axiom. We will soon see that the light
axiom can be regarded as the key assumption of
special relativity.

Since we want to axiomatize special relativity, we have to fix some
formal language in which we will write up our axioms.
Let us see the basic concepts
(the ``vocabulary" of the FOL language) we will use. We would
like to speak about motion. So we need a basic concept of things
that can move. We will call these object
\textit{bodies}.\footnote{By bodies we mean anything
  which can move, e.g., test-particles, reference frames,
  electromagnetic waves, etc.}  The light axiom requires a
distinguished type of bodies called \textit{photons} or \textit{light
  signals}.\footnote{Here we use light signals and photons as synonyms
  because it is not important here whether we think of them as
  particles or electromagnetic waves. The only thing that matters here
  is that they are ``things that can move.'' So they are bodies in the
  sense of our FOL language.}  We will represent motion as the
changing of spatial location in time.  Thus we will use reference
frames for coordinatizing events (meetings of bodies). Time and space
will be marked by \textit{quantities}.  The structure of quantities
will be an \textit{ordered field} in place of the field of real
numbers.\footnote{ Using ordered fields in place of the field of real
  numbers increases the flexibility of the theory and reduces the
  amount of mathematical presuppositions.  For further motivation in
  this direction, see, e.g., \citep{ax}.  Similar remarks apply to
  our other flexibility-oriented decisions, e.g., to treat the
  dimension of spacetime as a variable.}  For simplicity, we will associate special
  bodies to reference frames. These special bodies will be called
  ``observers.'' Observations will be formalized/represented
   by means of the {\it worldview relation}.

To formalize the ideas above, let us fix a natural number $d\ge 2$
for the dimension of spacetime. To axiomatize
theories of the $d$-dimensional spacetime, we will use the following
two-sorted FOL language:
\begin{equation*}
\{\, \B, \IOb, \Ph,\; \Q,+,\cdot,\; \W\,\},
\end{equation*}
where $\B$ (bodies) and $\Q$ (quantities) are the two
sorts,\footnote{That our theory is two-sorted
means only that there
  are two types of basic objects (bodies and quantities) as
  opposed to, e.g., set theory where there is only one type of basic
  objects (sets).}  $\IOb$ (inertial observers) and $\Ph$ (light
signals or photons) are one-place relation symbols of sort $\B$, $+$
and $\cdot$ are two-place function symbols of sort $\Q$, and $\W$ (the
worldview relation) is a $2+d$-place relation symbol the first two
arguments of which are of sort $\B$ and the rest are of sort $\Q$.

Atomic formulas $\IOb(k)$ and $\Ph(p)$ are translated as ``\textit{$k$
  is an inertial observer},'' and ``\textit{$p$ is a photon},''
respectively.  To speak about coordinatization, we translate
$\W(k,b,x_1,\ldots,x_{d-1},t)$ as ``\textit{body $k$ coordinatizes
  body $b$ at space-time location $\langle
  x_1,\ldots,x_{d-1},t\rangle$},'' (i.e., at space location $\langle
x,\ldots,x_{d-1}\rangle$ and at instant $t$). Sometimes we use the
more picturesque expressions \textit{sees} or \textit{observes} for
{\em coordinatizes}.  However, these cases of ``seeing'' and
``observing'' have nothing to do with visual seeing or observing;
they only mean associating coordinate points to
bodies.

The above, together with statements of the form $x=y$ are the
so-called \textit{atomic formulas} of our FOL language, where $x$ and
$y$ can be arbitrary variables of the same sort, or terms built up
from variables of sort $\Q$ by using the two-place operations $\cdot$
and $+$. The \textit{formulas} are built up from these atomic formulas
by using the logical connectives \textit{not} ($\lnot$), \textit{and}
($\land$), \textit{or} ($\lor$), \textit{implies} ($\rightarrow$),
\textit{if-and-only-if} ($\leftrightarrow$) and the quantifiers
\textit{exists} ($\exists$) and \textit{for all} ($\forall$).  For the
precise definition of the syntax and semantics of FOL, see, e.g.,
\cite[\S 1.3]{CK}.

To meaningfully formulate the light axiom, we have to provide some
algebraic structure for the quantities. Therefore,  in our
first axiom, we state some usual properties of
addition $+$ and multiplication $\cdot$ true for real numbers.
\begin{description}
\item[\underline{\ax{AxFd}}] The quantity part $\langle
  \Q,+,\cdot\rangle$ is a Euclidean field, i.e.,\\
$\bullet$ $\langle\Q,+,\cdot\rangle$ is a field in the sense of abstract
algebra,\\
$\bullet$ the relation $\le$ defined by
$\,x\le y\defiff \exists
  z\enskip x+z^2=y\,$ is a linear ordering on $\Q$, and\\
$\bullet$ Positive elements have square roots: $\forall x\enskip\exists y\enskip x=y^2\lor-x=y^2$.
\end{description}
\noindent The field-axioms (see, e.g., \cite[40{\ndash}41.]{CK})
say that $+$, $\cdot$ are associative and
commutative, they have neutral elements $0$, $1$ and inverses $-$,
$\oszt$ respectively, with the exception that $0$ does not have an
inverse with respect to $\cdot\,$, as well as $\cdot$ is additive with
respect to $+$. We will use $0$, $1$, $-$, $\oszt$, $\gyok$ as derived
(i.e., defined) operation symbols. 

\ax{AxFd} is a ``mathematical" axiom in spirit. However, it has
physical (even empirical) relevance. Its physical relevance is that
we can add and multiply the outcomes of our measurements and some
basic rules apply to these operations. Physicists usually use
all properties of the real numbers tacitly, without stating
explicitly which property is assumed and why. The two properties of
real numbers which are the most difficult to defend from an empirical
point of view are the Archimedean property, see \citep{Rosinger08},
\cite[\S 3.1]{Rosinger09}, and the supremum property,\footnote{The
supremum property (i.e., \textit{every} nonempty and
  bounded \textit{subset} of the real numbers has a least upper bound) implies
  the Archimedean property. So if we want to get ourselves free from
  the Archimedean property, we have to leave this property, too.} see
the remark after the introduction of axiom \ax{Cont} on
p.\pageref{p-cont}.

Euclidean fields got their name after their role in Tarski's FOL
axiomatization of Euclidean geometry \citep{TarskiElge}. By \ax{AxFd}
we can reason about the Euclidean structure of a coordinate system
the usual way, we can introduce Euclidean distance, speak about
straight lines, etc.  In particular, we will use the following
notation for $\vx,\vy\in\Q^n$ (i.e., $\vx$ and $\vy$ are $n$-tuples
over $\Q$) if $n\ge 1$:
\begin{equation*}
|\vx|\de\sqrt{x_1^2+\dots+x_n^2},\quad\text{ and }\quad
  \vx-\vy\de\langle x_1-y_1,\dots,x_n-y_n\rangle.
\end{equation*}

\noindent
We will also use the following two notations:
\begin{equation*}
\vx_s\de \langle x_1,\ldots, x_{d-1}\rangle \quad \text{ and }\quad {x_t}\de x_d
\end{equation*}
for the \textit{space component} and the
\textit{time component} of $\vx= \langle x_1,\ldots, x_d\rangle \in\Q^d$,
respectively.

Now let us see how the light axiom can be formalized in our FOL
language.

\begin{description}
\item[\underline{\ax{AxPh}}] For any inertial observer, the speed of
  light is the same in every direction everywhere, and it is
  finite. Furthermore, it is possible to send out a light signal in
  any direction. Formally:
\begin{multline*}
\forall m\enskip\exists c_m\enskip\forall \vx\vy \enskip \IOb(m)\rightarrow\\
\big(\exists p \enskip \Ph(p)\land \W(m,p,\vx)\land \W(m,p,\vy)\big)
\leftrightarrow |\vy_s-\vx_s|= c_m\cdot|y_t-x_t|.
\end{multline*}
\end{description}

Axiom \ax{AxPh} has an immediate physical meaning.  This axiom
is not only implied by the two original
principles of relativity, but it is well supported by
experiments, such as the Michelson-Morley
experiment. Moreover, it has been continuously tested ever since
then. Nowadays it is tested by GPS technology.

Axiom \ax{AxPh} says that ``It is \textit{possible} for a photon to
move from $\vx$ to $\vy$ iff ...''. So, a notion of possibility plays a
role here. In the present paper we work in an extensional framework,
as is customary in geometry and in spacetime theory.  However, it
would be more natural to treat this ``possibility phenomenon" in a
modal logic framework, and this is more emphatically so for
relativistic dynamics \citep{dyn-studia}. It would be interesting to
explore the use of modal logic in our logical analysis of relativity
theory. This investigation would be a nice unification of the works of
Imre Ruzsa's school on modal logic and the works of our Tarskian
spirited school on axiomatic foundations of relativity theory. Robin
Hirsch's work can be considered as a first step along this
road \citep{Hirsch}.

Let us note that \ax{AxPh} does not require that the speed of light be
the same for every inertial observer or that it be nonzero.  It
requires only that the speed of light according to a fixed inertial
observer be a quantity which does not depend on the direction or the
location.

Why do we not require that the speed of light is nonzero? The main
reason is that we are building our logical foundation of spacetime
theories examining thoroughly each part
of each axiom to see where and why we should assume them. Another
(more technical) reason is that it will be more natural to include
this assumption ($c_m\ne 0$) in our auxiliary axiom \ax{AxSm}
on page~\pageref{axsymd}.

Our next axiom connects the worldviews of different inertial observers
by saying that all observers observe the same ``external" reality (the
same set of events). Intuitively, by the event occurring for $m$ at
$\vx$, we mean the set of bodies $m$ observes at $\vx$. Formally:
\begin{equation*}
\ev_m(\vx)\de\{ b : \W(m,b,\vx)\}.
\end{equation*}

\begin{description}
\item[\underline{\ax{AxEv}}] All inertial observers coordinatize the
  same set of events:
\begin{equation*}
\forall mk\enskip \IOb(m)\land\IOb(k)\enskip\rightarrow\enskip
\forall \vx\enskip \exists \vy\enskip \forall b\enskip
\W(m,b,\vx)\leftrightarrow\W(k,b,\vy).
\end{equation*}
\end{description}

This axiom is very natural and tacitly assumed in the non-axiomatic
approaches to special relativity, too.

Basically we are done. We have formalized the light axiom
\ax{AxPh}. We have introduced two supporting axioms (\ax{AxFd} and
\ax{AxEv}) for the light axiom which are simple and natural; however,
we cannot simply omit them without loosing some of the meaning of
\ax{AxPh}. The field axiom enables us to speak about distances, time
differences, speeds, etc. The event axiom ensures that different
inertial observers see the same events.

In principle, we do not need more axioms for analyzing/axiomatizing
special relativity, but let us introduce two more simplifying ones.
We could leave them out without loosing the essence of our theory,
it is just that the formalizations of the theorems would become more
complicated.

\begin{description}
\item[\underline{\ax{AxSf}}] Any inertial observer sees himself
  on the time axis:
\begin{equation*}
\forall m\enskip\IOb(m)\rightarrow\enskip \big(\forall \vx\enskip
\W(m,m,\vx) \leftrightarrow x_1=0\land x_2=0\land x_3=0\big).
\end{equation*}
\end{description}

The role of \ax{AxSf} is nothing more than making it easier to speak
about the motion of reference frames via the motion of their time
axes. Identifying the motion of reference frames with the motion of
their time axes is a standard simplification in the literature.
\ax{AxSf} is a way to formally capture this simplifying
identification.

Our last axiom is a symmetry axiom saying that all inertial observers use the
same units of measurements.

\begin{description}
\item[\underline{\ax{AxSm}}]\label{axsymd}
Any two inertial observers agree about the spatial distance between
two events if these two events are simultaneous for both of them;
furthermore, the speed of light is 1:
\begin{multline*}
\forall mk\enskip\IOb(m)\land\IOb(k)\rightarrow \forall
\vx\vy\vx'\vy' \enskip x_t=y_t\land x'_t=y'_t\land\\
\ev_m(\vx)=\ev_k(\vx')\land \ev_m(\vy)=\ev_k(\vy')\rightarrow
|\vx_s-\vy_s|=|\vx'_s-\vy'_s|, \text{ and }
\end{multline*}
\begin{multline*}
\forall
m\enskip\IOb(m)\rightarrow\exists
p\enskip\Ph(p)\land\W(m,p,0,0,0,0)\land\W(m,p,1,0,0,1).
\end{multline*}
\end{description}

Let us see how \ax{AxSm} states that ``all inertial observers use
the same units of measurements.'' That ``the speed of light is 1''
(besides that the speed of light is nonzero) means only that
observers are using units measuring time distances compatible with
the units measuring spatial distances, such as light years or light
seconds. The first part of \ax{AxSm} means that different observers
use the same unit measuring spatial distances. This is so because
if two events are simultaneous for both observers, they can measure
their spatial distance and the outcome of their
measurements are the same iff the two observers are using the same
units to measure spatial distances.

\noindent
Our axiom system for special relativity contains these 5 axioms only:
\begin{equation*}
\ax{SpecRel} \de \{ \ax{AxFd}, \ax{AxPh}, \ax{AxEv}, \ax{AxSf},
\ax{AxSm}\}.
\end{equation*}

In an axiom system, the axioms are the ``price" we pay, and the
theorems are the ``goods" we get for them. Therefore, we strive for
putting only simple, transparent, easy-to-believe statements in our
axiom systems. We want to get all the hard-to-believe predictions as
theorems. For example, we prove from \ax{SpecRel} that it is
impossible for inertial observers to move faster than light relative
to each other (``No FTL travel" for science fiction fans). In the
following, $\vdash$ means logical derivability.

\begin{thm}{\rm (no faster than light inertial observers)} \label{thm-noftl}
\begin{multline*}
\ax{SpecRel}\vdash \forall mk\vx\vy \quad
\IOb(m)\land\IOb(k)\\\land\W(m,k,\vx)\land\W(m,k,\vy)\land\vx\neq\vy
\then |\vy_s-\vx_s|<|y_t-x_t|.
\end{multline*}
\end{thm}
\noindent
For a geometrical proof of Thm.\ref{thm-noftl}, see \citep{AMNSz-Synthese}.

In relativity theory we are often interested in comparing
the worldviews of different observers. So we introduce the worldview
transformation between observers $m$ and $k$ as the following binary
relation:
\begin{equation*}
\w_{mk}(\vx,\vy)\defiff \ev_m(\bar
x)=\ev_k(\vy).
\end{equation*}

By Thm.\ref{thm-poi}, the worldview transformations between inertial
observers in the models of \ax{SpecRel} are Poincar{\'e}
transformations, i.e., transformations which preserve the so-called
Minkowski-distance $(y_t-x_t)^2 - |\vy_s-\vx_s|^2$ of $d$-tuples
$\vy,\vx$. For the definition, we refer to \cite[110.]{dinverno} or
\cite[66{\ndash}69.]{MTW}.

\begin{thm}\label{thm-poi}
\begin{equation*}
\ax{SpecRel} \vdash \forall m,k \enskip \IOb(m)\land \IOb(k) \then \w_{mk}
\text{ \rm is a Poincar{\'e} transformation.}
\end{equation*}
\end{thm}

For the proof of Thm.\ref{thm-poi}, see \cite[Thm.11.10,
  640.]{logst} or \cite[Thm.3.2.2, 22.]{Szphd}.  By
Thm.\ref{thm-poi}, all predictions of special relativity, such as
``moving clocks slow down,'' are provable from \ax{SpecRel}.  For
details, see, e.g., \cite[\S 1]{AMNsamples}, \cite[\S 2]{logst},
\cite[\S 2.5]{pezsgo}.

\section{logical analysis}\label{anal-sec}

Let us illustrate here by a simple example what
we mean by logical analysis
of a theory. In \ax{AxEv} we have assumed that all observers
see the same (possibly infinite) meetings
of bodies. Let us try to weaken \ax{AxEv} to an axiom assuming
something similar but only for finite meetings
of bodies. A natural candidate is one of the
following finite approximations of \ax{AxEv}:
\begin{description}
\item[\underline{\ax{AxMeet_n}}] All inertial observers see
  the same  $n$-meetings of bodies:
\begin{multline*}
\forall mkb_1\ldots b_n\vx\enskip
\IOb(m)\land\IOb(k)\land\W(m,b_1,\vx)\land\ldots\land\W(m,b_n,\vx)\\
\then \exists \vy\enskip \W(k,b_1,\vy)\land\ldots\land\W(k,b_n,\vy).
\end{multline*}
\end{description}

For example, \ax{AxMeet_1} means only that inertial observers see
the same bodies.  Let us also introduce axiom scheme
\ax{Meet_\omega} as the collection of all the axioms \ax{AxMeet_n}.
By Prop.\ref{prop-meet}, \ax{AxMeet_n} is
strictly weaker assumption than \ax{AxMeet_{n+1}}
and \ax{AxEv} is strictly stronger than all the axioms of
\ax{Meet_\omega} together.

\begin{prop}\label{prop-meet}
\begin{eqnarray}
\ax{AxEv}&\vdash&\ax{AxMeet_{n+1}}\vdash\ax{AxMeet_{n}}\label{item-em}\\
\ax{AxMeet_n}&\nvdash&\ax{AxMeet_{n+1}}\label{item-mnm}\\
\ax{Meet_\omega}&\nvdash&\ax{AxEv} \label{item-mne}
\end{eqnarray}
\end{prop}

\begin{proof}
Item \eqref{item-em} follows easily by the formulations of the axioms.

\medskip
To prove Item \eqref{item-mnm}, we are going to construct a model of
\ax{AxMeet_n} in which \ax{AxMeet_{n+1}} is not valid. Let 
  $\B=\{b_i:i\le n\}$.Let
all the bodies be inertial observers.  Let $b_0$ see all the bodies
in $\langle 0,\ldots,0\rangle$ and none of them in any other
coordinate points, i.e., let $\W(b_0,b_i,\vx)$ hold iff $\vx=\langle
0,\ldots,0\rangle$; and for all $k\neq 0$ let $b_k$ see all the
bodies but $b_i$ at coordinate points $\langle i, \ldots, i \rangle
$ for all $i\le n$, i.e., let $\W(b_k,b_i,\vx)$
hold iff $\vx=\langle j,\ldots,j\rangle$ and $i\neq j$. In this
model, all inertial observers see all the possible $n$-meetings. So
\ax{AxMeet_n} is valid in this model. However, the only inertial
observer who sees the $n+1$-meeting $\{b_0,\ldots,b_n\}$ is $b_0$.
So \ax{AxMeet_{n+1}} is not valid in this model.

\medskip
We are going to prove Item \eqref{item-mne}
by a similar model construction. The only difference is that now
$\Q$ will be infinite. For simplicity, let $\Q$ be the set of
natural numbers. Let all the other parts of the
model be defined in the same way. Now all the inertial
observers see all the possible $n$-meetings of
the bodies for all natural numbers $n$. So
\ax{AxMeet_n} is valid in this model for all natural number $n$.
Hence \ax{Meet_\omega} is valid in this model. However, only $b_0$
sees the event $\{b_1, b_2, \ldots, \}$. So \ax{AxEv} is not valid in this
model.
\end{proof}

Now we will use that there are no stationary (i.e., motionless)
light signals. So let us formalize this statement.
\begin{description}
\item[\underline{\ax{Ax(c\neq0)}}] Inertial observers do not see
stationary light signals.
\begin{multline*}
\forall mp\vx\vy\quad
\IOb(m)\land\Ph(p)\land\W(m,p,\vx)\land\W(m,p,\vy)\land x_t\neq y_t\then \vx_s\neq\vy_s.
\end{multline*}
\end{description}

\begin{prop}\label{prop-m3e}
\begin{eqnarray}
\ax{AxMeet_3},\ax{AxFd},\ax{AxPh}, \ax{Ax(c\neq 0)} \vdash  \ax{AxEv}\label{item-m3e}\\
\ax{AxMeet_2},\ax{AxFd},\ax{AxPh}, \ax{Ax(c\neq 0)} \nvdash \ax{AxEv}\label{item-m2ne}\\
\ax{Meet_\omega},\ax{AxFd},\ax{AxPh} \nvdash \ax{AxEv}\label{item-mwne}
\end{eqnarray}
\end{prop}

\begin{proof}
First let us make some general observations.  By \ax{AxFd}, there is
no nondegenerate triangle in $Q^d$ whose sides are of slope $c$.
This is clear if $c=0$; and in the case $c\neq0$,
this can be shown by contradiction using the fact that
the vertical projection of a triangle of this kind is a triangle
whose one side is the sum of the other two sides. Therefore,
\ax{AxFd} and \ax{AxPh} together imply that any
inertial observer $m$ sees the events in which a particular photon
participates on a line of slope $c_m$.

By \ax{AxFd}, \ax{AxPh} and \ax{Ax(c\neq0)}, every inertial observer
$m$ sees different meetings of photons at different coordinate
points. This is so since (by \ax{AxFd}) for every
pair of points there is a line of slope $c_m\neq0$
containing only one of the points. Hence, by \ax{AxPh}, there is a
photon seen by $m$ only at one of the two coordinate points.

Let us now prove Item \eqref{item-m3e}.
Let $m$ and $k$ be inertial observers and let $\vx$ be a coordinate
point.  To prove \ax{AxEv}, we have to find a coordinate point $\vx'$
such that $\ev_m(\vx)=\ev_k(\vx')$. To find this $\vx'$, let
$\vy=\langle x_1+c_m,x_2,\ldots,x_{d-1},x_t+1\rangle$, $\vz=\langle
x_1-c_m,x_2,\ldots,x_{d-1},x_t+1\rangle$ and $\vw=\langle
x_1,\ldots,x_{d-1},x_t+2\rangle$, see Fig.\ref{fig-meet}.

By \ax{AxPh}, there are photons $p_1$, $p_2$ and
$p_3$ such that $p_1,p_2\in\ev_m(\vx)$, $p_2,p_3\in \ev_m(\vy)$,
$p_1\in\ev_m(\vz)$ and $p_3\in\ev_m(\vw)$. Since $m$ sees every
photon on a line of slope $c_m$, he sees the meeting of $p_1$ and
$p_2$ only at $\vx$ and does not see the meeting of $p_1$ and $p_3$.

Since \ax{AxMeet_3} implies \ax{AxMeet_2}, $k$ sees the same meetings
of pairs of photons. So there is a $\vx'$ where $k$ sees $p_1$ and
$p_2$ meet. $\vx'$ is the only point where $k$ sees both $p_1$ and
$p_2$. This is so because $k$ sees different meetings of photons at
different points but sees the same $3$-meetings as $m$. So if there
were another point, say $\vx''$, where $k$ sees $p_1$ and $p_2$, there
were photons $p'\in ev_k(\vx')$ and $p''\in\ev_k(\vx'')$ such that
$p'\not\in ev_k(\vx'')$, $p''\not\in\ev_k(\vx')$ and $k$ does not see
the meeting of $p'$ and $p''$. By axiom \ax{AxMeet_3} $m$ has to see
the meetings $\{p_1,p_2,p'\}$ and $\{p_1,p_2,p''\}$. The only point
where $m$ can see these meetings is $\vx$ since $\vx$ the only point
where $m$ sees $p_1$ and $p_2$ meet.  Therefore $m$ sees the meeting
of $p'$ and $p''$ at $\vx$. Thus, by \ax{AxMeet_3}, $k$ also has to
see the meeting of $p'$ and $p''$, but $k$ does not see it. Hence
$\vx'$ is the only point where $k$ sees both $p_1$ and $p_2$.

Let $b$ be a body such that $\W(m,b,\vx)$. By \ax{AxMeet_3}, $k$ has
to see the meeting of $p_1$, $p_2$ and $b$. This
point has to be $\vx'$ since the only point where $p_1$ and $p_2$
meet is $\vx'$. Since $b$ was an arbitrary body, we have
$\ev_m(\vx)\subseteq\ev_k(\vx')$. The same argument shows that
$\ev_k(\vx')\subseteq\ev_m(\vx)$. So $\ev_m(\vx)=\ev_k(\vx')$ as
desired.

\begin{figure}
\psfrag{m}[l][l]{$m$} \psfrag{k}[l][l]{$k$} \psfrag{p}[br][br]{$p$}
\psfrag{p1}[tl][tl]{$p_1$} \psfrag{p2}[l][l]{$p_2$}
\psfrag{p3}[l][l]{$p_3$} \psfrag{x}[tl][tl]{$\vx$}
\psfrag{y}[br][br]{$\vy$}
\psfrag{z}[b][b]{$\vz$} \psfrag{w}[bl][bl]{$\vw$}
\psfrag{x'}[tl][tl]{$\vx'$}
\psfrag{b}[tl][tl]{$b$}
\psfrag{1}[tl][tl]{$1$} \psfrag{cm}[tl][tl]{$c_m$}

\includegraphics[keepaspectratio, width=\textwidth]{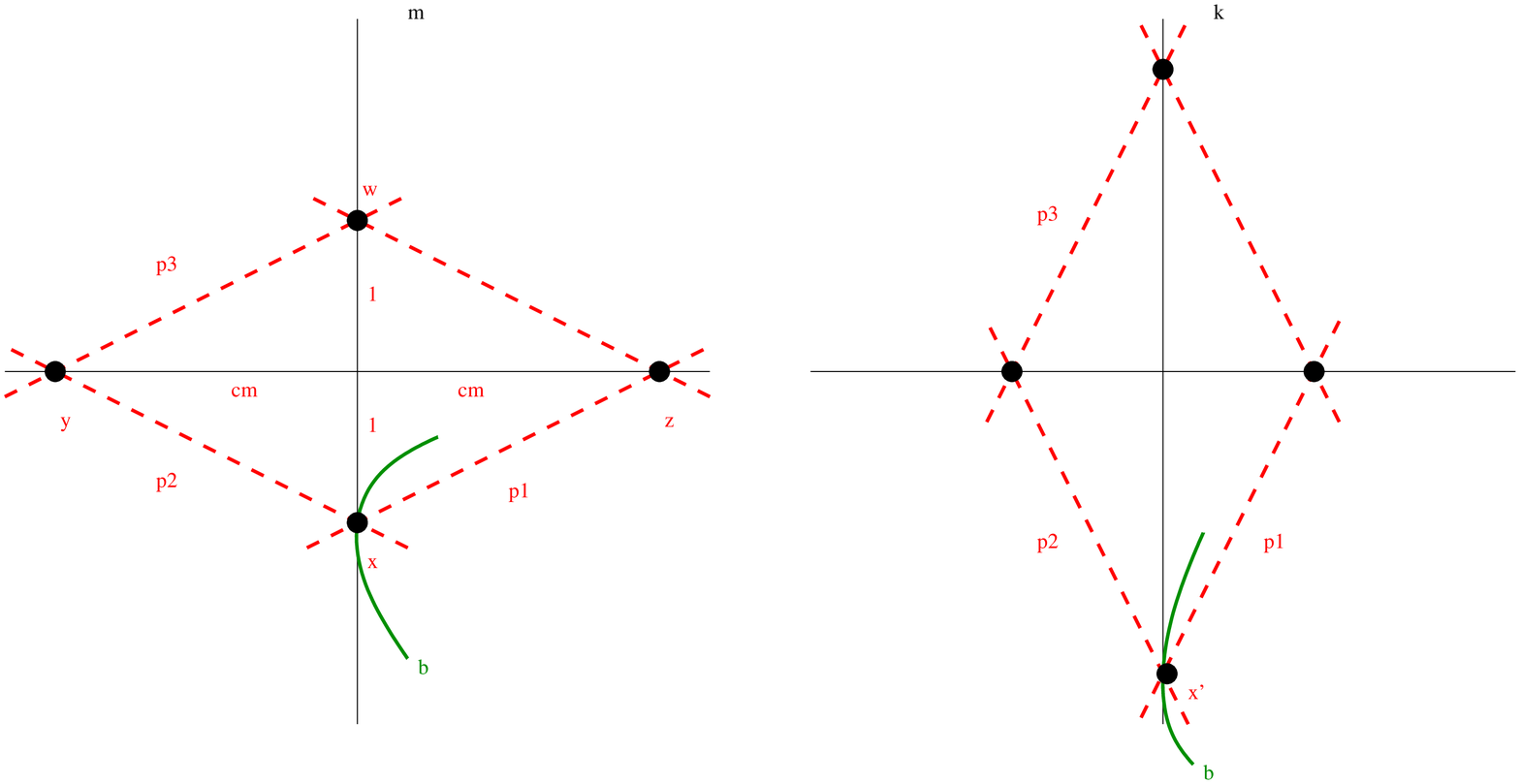}
\caption{\label{fig-meet}}
\end{figure}

\medskip
We are going to prove Item \eqref{item-m2ne}, by constructing a
model. Let $\langle\Q, +,\cdot\rangle$ be the field of real numbers.
Let us denote the set of natural numbers by $\omega$. Let
$\B=\{m,k\}\cup\{b_i:i\in\omega\}\cup\{p : p \text{ is a line of slope
  1}\}$. Let $m$ and $k$ be all the inertial observers and let the
lines of slope $1$ be all the photons. Let $m$ and $k$ see the photon
$p$ at coordinate point $\vx$ iff $\vx\in p$.  Let $m$ see all the
bodies $b_i$ at $\vx$ iff $x_t=0$. Let $k$ see all the bodies
$b_0,\ldots,b_n,\ldots$ but $b_i$ at $\vx$ iff $x_t=i$ (i.e., iff
$\vx$ is in the horizontal hyperplane $\{\vy\in\Q^d:
y_t=i\}$).\footnote{If $d=2$, vertical lines can be used instead of
  horizontal hyperplanes, which gives a counterexample with bodies
  having more natural properties.} It is straightforward from this
construction that axioms \ax{AxFd}, \ax{AxPh} and \ax{Ax(c\neq 0)} are
valid in this model.  Since every line of slope 1 intersects every
horizontal hyperplane, $m$ and $k$ see the same $2$-meetings of
bodies. Hence \ax{AxMeet_2} is also valid in this model.  However, the
only inertial observer who sees
the meeting $\{b_i:i\in\omega\}$ is $m$. So \ax{AxEv} is not valid in
this model.

\medskip
We prove Item \eqref{item-mwne} by a
similar construction.  The only difference is that now the set of
bodies is $\B=\{m,k\}\cup\{b_i:i\in\omega\}\cup\{p : p \text{ is a
vertical
  line}\}$; and the photons are the vertical lines. It is
straightforward from the construction that axioms \ax{AxFd},
\ax{AxPh} are valid in this model ($c=0$).  Since every vertical
line intersects every horizontal hyperplane, $m$ and $k$ see
the same $n$-meetings of bodies. Hence \ax{Meet_\omega} is also
valid in this model. However, only $m$ sees
the meeting $\{b_i:i\in\omega\}$. So \ax{AxEv} is not valid in this model.
\end{proof}

Prop.\ref{prop-m3e} shows that a price to weaken axiom \ax{AxEv} to
\ax{AxMeet_3} is to assume that there are no stationary light
signals. Since \ax{AxSm} contains this assumption, we can simply
replace \ax{AxEv} with \ax{AxMeet_3} in \ax{SpecRel}. A natural
continuation of this investigation can be a search for assumptions
that allow us to weaken \ax{AxMeet_3} to \ax{AxMeet_2}. A possible
candidate is that bodies move along straight lines and the dimension $d$ is at
least $3$. The proof of Item \eqref{item-m2ne} shows that assuming
only that bodies move along straight lines is not enough, if $d=2$.

We have several similar investigations on the logical connections of
axioms and predictions, see, e.g., \citep{dyn-studia}, \cite[\S
  5]{Szphd} on dynamics, \citep{twp}, \cite[\S 4,\S 7 ]{Szphd},
\citep{gcocp} on twin paradox, \citep{pezsgo} on kinematics,
time-dilation and length-contraction, twin paradox, etc.

\section{accelerated observers}\label{acc-sec}

In $\ax{SpecRel}$ we restricted our attention to inertial observers.
It is a natural idea to generalize the theory by including
accelerated observers as well.  It is explained in the classic
textbook \cite[163{\ndash}165.]{MTW} that the study of accelerated
observers is a natural first step (from special
relativity) towards general relativity.

We have not introduced the concept of observers as a basic one because
it can be defined as follows: an {\it observer} is nothing other than a
body who ``observes" (coordinatizes) some other bodies somewhere, this
property can be captured by the following formula of our language:
\begin{equation*}
\Ob(m) \defiff  \exists b\vx\enskip \W(m,b,\vx).
\end{equation*}

Our key axiom about accelerated observers is the following:
\begin{description}
\item[\ax{AxCmv}] At each moment of his life, every accelerated
  observer sees (coordinatizes) the nearby world for a short while in
  the same way as an {inertial observer} does.
\end{description}
For formulation of \ax{AxCmv} in our FOL language, see \citep{twp},
\citep{Szphd} or \citep{AMNSz-Synthese}.

Axiom \ax{AxCmv} ties the behavior of accelerated observers to those
of inertial ones. Justification of this axiom is given by
experiments. We call two observers \textit{co-moving} at an event if
they ``see the nearby world for a short while in the same way'' at
the event. By this notion \ax{AxCmv} says that at each event of an
observer's life, he has a co-moving inertial observer. We can think
of a dropped spacepod as a co-moving inertial observer of an
accelerated spaceship (at the event of dropping). Or, if a spaceship
switches off its engines, it will move on as a co-moving inertial
spaceship would.

Our next two axioms ensure that the worldviews of accelerated
observers are big enough. They are generalized versions of the
corresponding axioms for inertial observers, but now postulated for
all observers.

\begin{description}
\item[\underline{$\ax{AxEv^-}$}] If $m$ sees $k$ in an event, then $k$
cannot deny it:
\begin{equation*}\forall m,k\in\Ob\enskip  \W(m,k,\vx)\rightarrow\exists
\vy\enskip \ev_m(\vx)=\ev_k(\vy).
\end{equation*}
\item[\underline{$\ax{AxSf^-}$}]  Any observer  sees  himself in  an
  interval of the time axis:
\begin{multline*}
 \forall m\in\Ob\enskip\forall \vx\enskip \W(m,m,\vx)\then x_1=x_2=x_3=0
 \quad \text{ and }\\
\forall \vx\vy\quad \W(m,m,\vy)\land\W(m,m,\vx)\then \forall t\enskip
x_t<t<y_t\rightarrow  \W(m,m,0,0,0,t).
\end{multline*}
\end{description}

Our last two axioms will ensure that the worldlines of accelerated
observers are ``tame" enough, e.g., they have velocities at each
moment. In \ax{SpecRel}, the worldview transformations between inertial
observers are affine maps, the next axiom will state that the
worldview transformations between accelerated observers are
approximately affine, wherever they are defined.

\begin{description}
\item[\underline{\ax{AxDf}}] The worldview transformations have
linear approximations at each point of their domain (i.e., they
are differentiable).
\end{description}
For a precise formalization of \ax{AxDf}, see, e.g., \citep{AMNSz-Synthese}.

We note that \ax{AxDf} implies that the worldview transformations are
functions with open domains. However, if the numberline has gaps,
still there can be crazy motions. Our last assumption is an axiom
scheme supplementing \ax{AxDf} by excluding these gaps.

\begin{description}
\label{p-cont}
\item[\underline{\ax{Cont}}] Every definable, bounded and nonempty
  subset of $\Q$ has a supremum (i.e., least upper bound).
\end{description}
In \ax{Cont} ``definable'' means ``definable in the language of
\ax{AccRel}, parametrically.'' For a precise formulation of \ax{Cont},
see \cite[692.]{twp} or \cite[\S 10.1]{Szphd}.
\ax{Cont} is a ``mathematical axiom" in spirit. It is Tarski's FOL
version of Hilbert's continuity axiom in his axiomatization of
geometry, see \cite[61{\ndash}162.]{Gol}, fitted to the language of
\ax{AccRel}.  When $\Q$ is the field of real numbers, \ax{Cont} is
automatically true.

That \ax{Cont} requires the existence of supremum only for sets
definable in the language of \ax{AccRel} instead of every set, is
important not only because by this trick we can keep our theory
within FOL (which is crucial in a foundational work), but also
because it makes this postulate closer to the the physical/empirical
level. The latter is true because \ax{Cont} does not speak
about ``any fancy subset'' of the quantities, just those
``physically meaningful'' sets which can be defined in the language
of our (physical) theory.

Adding this 5 axioms to \ax{SpecRel}, we get an axiom system for
accelerated observers:
\begin{equation*}
\ax{AccRel}\de\ax{SpecRel}\cup\{\ax{AxCmv}, \ax{AxEv^-}, \ax{AxSf^-},
\ax{AxDf}\}\cup\ax{Cont}.
\end{equation*}

\noindent As an example we show that the so-called \textit{twin
  paradox} can be naturally formulated and analyzed logically in
\ax{AccRel}.  Our axiomatic approach also makes it possible to analyze
the details of the twin paradox (e.g., who sees what, when) with the
clarity of logic, see \cite[139{\ndash}150.]{pezsgo} for part of such an
analysis.

According to the twin paradox, if a twin makes a journey into space
(accelerates), he will return to find that he has aged less than his
twin brother who stayed at home (did not accelerate). We formulate
the twin paradox in our FOL language as follows.

\begin{description}
\item[\underline{\ax{TwP}}]  Every inertial observer $m$
 measures at least as much time as any other observer $k$ between any
 two events $e_1$ and $e_2$ in which they meet; and they measure the same time iff they
 have encountered the very same events between $e_1$ and $e_2$:
\begin{multline*}
\forall m \in \IOb\enskip \forall k\in \Ob\enskip \forall
\vx\vx'\vy\vy' \enskip x_t<y_t \land x'_t<y'_t \,\land \\
m,k\in\ev_m(\vx)=\ev_k(\vx')\land m,k\in\ev_m(\vy)=\ev_k(\vy') \then
y'_t-x'_t\le y_t-x_t \\ \land\big( y'_t-x'_t=y_t-x_t \iff
enc_m(\vx,\vy)=enc_k(\vy',\vy') \big),
\end{multline*}
\end{description}
where  $enc_m(\vx,\vy)=\{\ev_m(\vz): \W(m,m,\vz) \land x_t\le z_t\le y_t\}$.

\begin{thm}\label{thm-twp}
\begin{eqnarray}
\ax{AccRel}\vdash\ax{TwP}\\
\ax{AccRel}-\ax{AxDf}\vdash\ax{TwP}\\
\ax{AccRel}-\ax{Cont}\nvdash\ax{TwP}\\
\ax{Th(\mathbb{R})}\cup\ax{AccRel}-\ax{Cont}\nvdash\ax{TwP}\label{item-nothr}
\end{eqnarray}
\end{thm}

For the proof of Thm.\ref{thm-twp}, see \citep{twp} or
\citep[\S 7]{Szphd}.

Item \eqref{item-nothr} of Thm.\ref{thm-twp} states that \ax{Cont}
cannot be replaced with the whole FOL theory of real numbers in
\ax{AccRel} if we do not want to loose \ax{TwP} from
its consequences.

Our theory $\ax{AccRel}$ is also strong enough to predict the
gravitational time-dilation effect of general relativity via
Einstein's equivalence principle, see \citep{MNSz}, \citep{Szphd}.

\section{general relativity}\label{gen-sec}

Our theory of accelerated observers \ax{AccRel} speaks about two kinds
of observers, inertial and accelerated ones. Some axioms are
postulated for inertial observers only, some apply to all
observers. We get an axiom system \ax{GenRel} for general relativity
by stating the axioms of \ax{AccRel} in a generalized form in which
they are postulated for all observers, inertial and accelerated ones
equally. In other words, we will change all axioms of \ax{AccRel} in
the same spirit as $\ax{AxSf^-}$ and $\ax{AxEv^-}$ were obtained from
\ax{AxSf} and \ax{AxEv}, respectively.
This kind of change \ax{AccRel} $\mapsto$ \ax{GenRel} can be
regarded as a ``democratic revolution'' with the slogan ``all
observers should be equivalent, the same laws should
apply to all of them.'' Here ``law'' translates as ``axiom.'' This
idea originates with Einstein (see his book \cite[Part II,
ch.18]{Ein}).

For simplicity, we will use an equivalent version of the symmetry
axiom \ax{AxSm} (see \cite[Thm.2.8.17(ii), 138.]{pezsgo} or
\cite[Thm.3.1.4, 21.]{Szphd}), and we will require the speed of
photons to be 1 in $\ax{AxPh^-}$ (as opposed to requiring it in
$\ax{AxSm^-}$).

\begin{description}
\item[\underline{$\ax{AxPh^-}$}] The velocity of photons an observer ``meets" is 1
when they meet, and it is possible to send out a photon in each
direction where the observer stands.
\item[\underline{$\ax{AxSm^-}$}] Meeting observers see each other's clocks slow
down with the same rate.
\end{description}
For a precise formulation of these axioms, see \citep{AMNSz-Synthese},
\citep{Szphd}.

We introduce an axiom system for general relativity as the collection
of the following axioms:
\begin{equation*}
\ax{GenRel}\de\{\ax{AxFd},\ax{AxPh^-},\ax{AxEv^-},\ax{AxSf^-},\ax{AxSm^-},\ax{AxDf}\}\cup\ax{Cont}.
\end{equation*}

Axiom system \ax{GenRel} contains basically the same axioms as
\ax{SpecRel}, the difference is that they are assumed only locally but
for all the observers.

Thm.\ref{grcomp-thm} below states that the models of
\ax{GenRel} are exactly the spacetimes of usual general relativity.
For the notion of a Lorentzian manifold we refer to
\cite[55.]{dinverno}, \cite[241.]{MTW} and \cite[sec.3.2]{logst}.

\begin{thm}[Completeness theorem]\label{grcomp-thm}
$\ax{GenRel} $ is complete with respect to its standard models,
i.e., with respect to
  Lorentzian Manifolds over real closed fields.
\end{thm}

This theorem can be regarded as a completeness theorem in the
following sense. Let us consider Lorentzian manifolds as intended
models of \ax{GenRel}. How can we do that? We give a method for
constructing a model of \ax{GenRel} from each Lorentzian manifold;
and conversely, we show that each model of \ax{GenRel} is obtained
this way from a Lorentzian manifold.  After this is elaborated, we
have defined what we mean by a formula $\varphi$ in the language of
\ax{GenRel} being valid in a Lorentzian manifold. Then completeness
means that for any formula $\varphi$ in the language of \ax{GenRel},
we have $\ax{GenRel}\vdash \varphi$ iff $\varphi$ is valid in all
Lorentzian manifolds over real closed fields. This is completely
analogous to the way in  which Minkowskian spacetimes were regarded
as intended models of \ax{SpecRel} in the completeness theorem of
\ax{SpecRel}, see \cite[Thm.11.28, 681.]{logst} and \cite[\S
4]{Mphd}.

We call the worldline of an observer \textit{timelike geodesic}, if
each of its points has a neighborhood within which this observer
``maximizes measured time (wrist-watch time)" between any two
encountered events. For formalization of this concept in our FOL
language, see, e.g., \citep{AMNSz-Synthese}.

According to the definition above, if there are only a few observers,
then it is not a big deal that a worldline is a time-like geodesic (it
is easy to be maximal if there are only a few to be compared to). To
generate a real competition for the rank of having a timelike geodesic
worldline, we postulate the existence of many observers by the
following axiom scheme of comprehension.

\begin{description}
\item[\underline{\ax{Compr}}] For any parametrically definable
  timelike curve in any observers worldview, there is another observer
  whose worldline is the range of this curve.
\end{description}

A precise formulation of \ax{Compr} can be obtained from that of its
variant in \cite[679.]{logst}.

An axiom schema \ax{Compr} guarantees that our definition of a
geodesic coincides with that in the literature on Lorentzian
manifolds. Therefore we also introduce the following theory:
\begin{equation*}
\ax{GenRel^+}\de \ax{GenRel}\cup\ax{Compr}.
\end{equation*}
So in our theory \ax{GenRel^+}, our concept of
timelike geodesic coincides with the standard
concept in the literature on general relativity.
All the other key concepts of general relativity,
such as curvature or Riemannian tensor field, are definable from
timelike geodesics. Therefore we can treat all these
concepts (including the concept
of metric tensor field) in our theory \ax{GenRel^+} in a natural
way.

In general relativity, Einstein's field equations (EFE) provide the
connection between the geometry of spacetime and the energy-matter
distribution (given by the energy-momentum tensor field). Since in
\ax{GenRel^+} all the geometric concepts of spacetime are definable,
we can use Einstein's equation as a definition of the
energy-momentum tensor, see, e.g., \citep{Benda} or \cite[\S 13.1,
169.]{dinverno}, or we can extend the language of \ax{GenRel^+}
with the concept of energy-momentum tensor and assume Einstein's
equations as axioms. As long as we do not assume anything more of
the energy-momentum tensor than its connection to the geometry
described by Einstein's equations, there is no real difference in
these two approaches. In both approaches, we can add extra
conditions about the energy-momentum tensor to our theory, e.g., the
dominant energy condition or, e.g., that the spacetimes are vacuum
solutions.

\section{can physics give feedback to logic?}

There is observational
evidence suggesting that in our physical
universe there exist regions supporting potential non-Turing
computations. Namely, it is possible to design a physical
device in relativistic spacetime which can compute a non-Turing
computable task, e.g., which can decide whether ZF set theory is
consistent. This empirical evidence is making the theory of
hypercomputation more interesting and gives new
challenges to the physical Church Thesis, see, e.g.,
\citep{ANN-natcomp}.

These new challenges do more than simply providing a further
connection between logic and spacetime theories; they also motivate
the need for logical understanding of spacetime theories.

\section{concluding remarks}

We have axiomatized both special and general relativity in FOL.
Moreover, via our theory \ax{AccRel}, we have axiomatized general
relativity so that each of its axioms can be traced
back to its roots in the axioms of special relativity.
Axiomatization is not our final goal. It is merely
an important first step toward logical and conceptual
analysis.  We are only at the beginning of our
ambitious project.\footnote[1]{This research is
  supported by the Hungarian Scientific Research Fund for basic
  research grant No.~T81188, as well as by a Bolyai grant
  for J.~X.~Madar\'asz.}

\bigskip\bigskip 
\noindent {\sc H.~Andr\'eka, J.~X.~Madar{\'asz}, \\
I. N\'emeti, G. Sz{\'e}kely\\} Alfr{\'e}d R{\'e}nyi
Institute of Mathematics\\ of the Hungarian Academy of Sciences\\Budapest P.O.
Box 127, H-1364 Hungary\\ {\tt andreka@renyi.hu, madarasz@renyi.hu\\
nemeti@renyi.hu,  turms@renyi.hu}


\begin{thebibliography}{}
\small

\bibitem{pezsgo} Andr{\'e}ka, H., J.~X.~Madar{\'a}sz, and
  I.~N{\'e}meti, with contributions from A.~Andai, G.~S{\'a}gi,
  I.~Sain and Cs.~T{\H o}ke, 2002, {\it On the logical structure of
    relativity theories}.  Research report. Budapest, Alfr{\'e}d R{\'e}nyi
  Institute of Mathematics.  http://www.renyi.hu/pub/algebraic-logic/Contents.html.

\bibitem{AMNsamples}
  Andr{\'e}ka, H., J.~X.~Madar{\'a}sz, and
  I.~N{\'e}meti, 2006, Logical axiomatizations of space-time. Samples from
  the literature. In A.~Pr{\'e}kopa, et~al. (Eds.) {\it Non-Euclidean
    geometries}. Berlin, Springer. 155{\ndash}185.

\bibitem{logst} Andr{\'e}ka, H., J.~X.~Madar{\'a}sz, and
  I.~N{\'e}meti, 2007, Logic of space-time and relativity theory. In
  M.~Aiello, et~al. (Eds.), {\it
    Handbook of Spatial Logics}. Berlin, Springer.
607{\ndash}711.

\bibitem{dyn-studia} Andr{\'e}ka, H., J.~X.~Madar{\'a}sz,
  I.~N{\'e}meti, and G.~Sz{\'e}kely, 2008, Axiomatizing relativistic
  dynamics without conservation postulates. {\it Studia Logica
  89}. 163{\ndash}186.

\bibitem{ANN-natcomp} Andr{\'e}ka, H., I.~N{\'e}meti, and
  P.~N{\'e}meti, 2009, General relativistic hypercomputing and foundation
  of mathematics. {\it Nat.~Comp. 8}. 499{\ndash}516.

\bibitem{AMNSz-Synthese} Andr{\'e}ka, H., J.~X.~Madar{\'a}sz,
  I.~N{\'e}meti, and G.~Sz{\'e}kely, 2010, A logic road from special
  relativity to general relativity. {\it Synthese}, submitted.

\bibitem{ax} Ax, J., 1978, \newblock The elementary foundations of
  spacetime.  \newblock {\em Found.~Phys. 8}. 507{\ndash}546.

\bibitem{Benda} Benda, T., 2008,  \newblock A formal construction of the
  spacetime manifold.  \newblock {\em J.~Philos~Logic 37}. 441{\ndash}478.

\bibitem{CK} Chang, C.~C., and H.~J.~Keisler, 1990, \newblock {\em Model
  theory}.  \newblock Amsterdam, {N}orth-{H}olland.

\bibitem{dinverno} d'Inverno, R., 1992, \newblock {\em Introducing
  {E}instein's relativity}. \newblock Oxford, {O}xford {U}niv. Press.

\bibitem{einstein05} Einstein, A., 1905/1952, \newblock {Zur Elektrodynamik
  bewegter K{\"o}rper}.  \newblock {\em Annalen der Physik.
  17}. 891{\ndash}921. English translation in A.~Einstein, {\em The
    principle of Relativity}. Mineola(NY), Dover.

\bibitem{Ein} Einstein, A., 1921/2006, \newblock {\em Relativity. The Special and
  the General Theory.}  \newblock London,  Penguin Classics. Translated
  by W. Lawson.



\bibitem{FriFOM1}
Friedman, H., 2004a, 
\newblock On foundational thinking 1.
\newblock Posting in FOM (Foundations of Mathematics) Archives
www.cs.nyu.edu (Jan.~20, 2004).

\bibitem{FriFOM2} Friedman, H., 2004b, \newblock On foundations of special
  relativistic kinematics 1. \newblock Posting in FOM (Foundations of
  Mathematics) Archives www.cs.nyu.edu (Jan.~21, 2004).


\bibitem{Gol} Goldblatt, R., 2004, \newblock {\em Orthogonality and spacetime
  geometry.} \newblock Berlin, Springer.

\bibitem{GSz}
G{\"o}m{\"o}ri, M., and L.~E.~Szab{\'o}, 2010,
\newblock {Is the relativity principle consistent with electrodynamics? Towards a
  logico-empiricist reconstruction of a physical theory}.
\newblock arXiv:0912.4388v3.


\bibitem{guts}
Guts, A.~K., 1982,
\newblock The axiomatic theory of relativity.
\newblock {\em Russ.~Math.~Surv. 37}. 41{\ndash}89.


\bibitem{herschbach} Herschbach, D., 2008, Einstein as a student. In
  P.~L.~Galison, et~al.
(Eds.)  {\em Einstein for the 21st century}. Princeton, Princeton Univ. Press. 217{\ndash}238.


\bibitem{Hirsch} Hirsch, R., 2009, Relativity and Modal Logic. {\em Hungarian Philosophical Review}, this issue.


\bibitem{Mphd} Madar{\'a}sz, J.~X., 2002, {\em Logic and Relativity: in the
  light of definability theory}.  PhD thesis, Budapest, E{\"o}tv{\"o}s
  Lor{\'a}nd Univ.

\bibitem{twp} Madar{\'a}sz, J.~X., I.~N{\'e}meti, and G.~Sz{\'e}kely, 2006,
  Twin paradox and the logical foundation of relativity theory.
  {\em Found.~Phys. 36}. 681{\ndash}714.


\bibitem{MNSz} Madar{\'a}sz, J.~X., I.~N{\'e}meti, and G.~Sz{\'e}kely, 2007,
  First-order logic foundation of relativity theories. In
  D.~Gabbay, et~al. (Eds.), {\em Mathematical problems from applied logic II}. Berlin, Springer. 217{\ndash}252.

\bibitem{MTW} Misner, C.~W., K.~S.~Thorne, and J.~A.~Wheeler, 1973,
  \newblock {\em Gravitation}.  \newblock New York, {W}. {H}. {F}reeman and
  Co.

\bibitem{Rosinger08} Rosinger, E.~E., 2008, \newblock {\em Two Essays on the Archimedean
  versus Non-Archimedean Debate}. \newblock  arXiv:0809.4509v3.

\bibitem{Rosinger09} Rosinger, E.~E., 2009, \newblock {\em Special Relativity
  in Reduced Power Algebras}. \newblock  arXiv:0903.0296v1.

\bibitem{schutz} Schutz, J.~W., 1973, \newblock {\em Foundations of special
  relativity: kinematic axioms for {M}inkowski space-time}.  \newblock
Berlin,  Springer.

\bibitem{suppes}
Suppes, P., 1968,
\newblock The desirability of formalization in science.
\newblock {\em J.~Philos. 27}. 651{\ndash}664.

\bibitem{Szabo}
Szab{\'o}, L.~E., 2009,
\newblock Empirical Foundation of Space and Time.
\newblock In M. Su{\'a}rez, et~al.
(Eds.), EPSA07: {\em Launch of the European Philosophy of Science Association}. Berlin, Springer.

\bibitem{Szphd} Sz{\'e}kely, G., 2009, {\em First-Order Logic
  Investigation of Relativity Theory with an Emphasis on Accelerated
  Observers}. PhD thesis, Budapest, E{\"o}tv{\"o}s
  Lor{\'a}nd Univ.

\bibitem{wqp} Sz{\'e}kely, G., 2010a, On Why-Questions in Physics.
  \newblock In {F}.~{S}tadler {\em et~al.},  (Ed.),  {\em {W}iener {K}reis und
    {U}ngarn}. Berlin, {S}pringer, to appear.

\bibitem{gcocp} Sz{\'e}kely, G., 2010b, A Geometrical Characterization
  of the Twin Paradox and its Variants. {\em Studia Logica 95}. pp.161--182.

\bibitem{TarskiElge} Tarski, A., 1959, \newblock What is elementary geometry?
  In {L}.~{H}enkin, et~al. (Eds.) {\em The axiomatic method. {W}ith
    special reference to geometry and physics.} \newblock
 Amsterdam,  North-Holland.16{\ndash}29.


\bibitem{vaananen}
V{\"a}{\"a}n{\"a}nen, J., 2001,
\newblock Second-order logic and foundations of mathematics.
\newblock {\em Bull.~Symb.~Log. 7}. 504{\ndash}520.


\bibitem{wolenski} Wole{\'n}ski, J., 2004, \newblock First-order logic:
  (philosophical) pro and contra.  \newblock In V.~F.~Hendricks
  et~al. (Eds.), {\em First-Order Logic Revisited}.
 Berlin,  Logos. 369{\ndash}398.


\end{thebibliography}
\end{document}